\documentclass[journal]{ieeeconf}                                                          
\IEEEoverridecommandlockouts

\usepackage{graphicx}
\usepackage{amssymb}
\usepackage{algorithmicx, algpseudocode}
\usepackage{comment}
\usepackage{algorithm}
\usepackage{algpseudocode}
\usepackage{amsmath} % assumes amsmath package installed
\usepackage{algorithmicx}
\renewcommand{\baselinestretch}{.985}

\newcommand{\oprocendsymbol}{\hbox{$\bullet$}}
\newcommand{\oprocend}{\relax\ifmmode\else\unskip\hfill\fi\oprocendsymbol}

\newtheorem{thm}{Theorem}[section]

\newtheorem{rem}{Remark}[section]

\newtheorem{lem}{Lemma}[section]

\newtheorem{assump}{Assumption}



\begin{document}\fontsize{10}{10.89}\rm

\title{Scalar Communication via Random Direction Refreshing for \\Distributed Optimization}

\author{Mohammadreza Rostami and Solmaz S. Kia, \emph{Senior member, IEEE} %
  \thanks{This work is supported by the 2026 UCLA–UCI grant awarded by the Samueli Foundation. The authors are with the Department of Mechanical and Aerospace Engineering, University of California Irvine, Irvine, CA 92697,  
    {\tt\small \{mrostam2,solmaz\}@uci.edu.}}%
}

\markboth{}%
{}
\maketitle
\begin{abstract}
Distributed optimization over networks requires agents to repeatedly exchange decision variables with their neighbors. When the decision dimension $d$ is large, these exchanges dominate the communication cost, which is critical for bandwidth-constrained agents. Existing remedies quantize or sparsify the exchanged vectors, yet each message still scales with $d$ and the compression error must be compensated by additional states. To address this limitation, we propose a scalar-communication mechanism in which every neighbor message carries a single real number regardless of $d$. Agents regenerate a common random direction from a shared seed, transmit only the inner product of their state with that direction, and act on the resulting rank-one surrogate of their neighbors' states while retaining full local gradients. We develop and analyze the mechanism for an existing continuous-time distributed optimization algorithm. For strongly convex local costs with Lipschitz gradients, we show that the optimizer remains the unique consensus equilibrium, that a fixed direction admits spurious equilibria, and that refreshing the direction at a sufficiently high rate yields exponential mean-square and almost-sure convergence with constant gains and no residual error. The framework admits any isotropic fixed-norm direction distribution, including Rademacher, scaled-coordinate, and sphere-normalized Gaussian directions; all three attain lower fresh-encoding variance than unnormalized Gaussian directions. The effects of the
direction distribution and the refresh interval are illustrated in~simulations.
\end{abstract}

\vspace{-0.2in}
\IEEEpeerreviewmaketitle

\section{Introduction}
\label{sec:introduction}
\vspace{-0.05in}
Distributed optimization enables a network of $N$ agents to
cooperatively solve
\begin{equation}\label{eq:global_opt}
  \min_{z\in\mathbb{R}^{d}} F(z)=\sum\nolimits_{i=1}^{N} f^i(z),
\end{equation}
where each local objective $f^i$ is private to agent~$i$, using
only local computation and communication with neighbors. Problems
of this form arise in cooperative control, distributed estimation,
resource allocation, and networked learning, and a mature set of
discrete- and continuous-time algorithms with strong convergence
guarantees is now available; see,
e.g.,~\cite{nedic2009distributed,qu2017harnessing,nedic2018distributed,yang2019survey}.

Despite these advances, communication remains a key impediment to
deploying distributed optimization on resource-constrained networks \cite{rostami2026adaptive}.
Communication cost has two distinct sources. The first is
\emph{energy}: in battery-powered systems, every transmission
depletes a limited budget, so agents should communicate only when
needed. The second is \emph{bandwidth}: the channel limits how much
information a single message can carry. This is a hard constraint
that holds regardless of battery status and can only be met by
reducing the size of each message. Event- and self-triggered
mechanisms~\cite{dimarogonas2012event,seyboth2013event,SSK-JC-SM:15optim}
address the energy side by reducing how often agents communicate,
but each message still contains a full $d$-dimensional state or
gradient. Quantization and sparsification reduce the number of bits
or entries sent, but the exchanged objects remain $d$-dimensional
and require auxiliary error-compensation or tracking states. These
limitations become severe when the decision variable is high
dimensional, $d\gg1$. This paper addresses the bandwidth constraint
by developing a distributed optimization algorithm whose per-message
cost is independent of~$d$.

Our approach builds on a scalar encoding principle from our recent
federated learning work~\cite{rostami2024fedscalar,rostami2026scalar, rostami2025fedmpdd}.
For any random direction $u\in\mathbb{R}^d$ with
$\mathbb{E}[uu^\top]=I_d$, a vector $g\in\mathbb{R}^d$ satisfies
\begin{equation}\label{eq:decomposition}
  g=\mathbb{E}\big[\,\underbrace{u}_{\text{vector}}\,
     \underbrace{u^\top g}_{\text{scalar}}\,\big],
\end{equation}
so a sender that shares a random seed with the receiver need only
transmit the scalar $u^\top g$, from which the receiver forms the
unbiased estimate $(u^\top g)\,u$. We apply this principle to the
states exchanged between neighbors in consensus-based distributed
optimization. All agents regenerate a common direction from a shared
seed, each agent transmits the projection of its state onto that
direction, and every agent keeps using its full local gradient. No
direction vector is transmitted, and no compression-error memory is
required.

\emph{Related Work:} In bandwidth-limited settings, the central
concern is the size of each transmitted message, and two major
approaches have been developed to address it. The first reduces the
number of bits used to represent each entry of the exchanged vector.
Building on quantized
consensus~\cite{kashyap2007quantized,carli2010gossip,li2011distributed},
the effect of quantization on distributed subgradient methods was
first characterized in~\cite{nedic2008quantization}. Subsequent work
established data-rate requirements for
convergence~\cite{yi2014quantized}, designed quantizers tailored to
optimization dynamics~\cite{pu2017quantization}, and analyzed
convergence under limited
communication~\cite{magnusson2018convergence}. With a fixed
quantizer, these methods generally converge to a neighborhood of the
optimizer, whereas adaptive quantization that progressively shrinks
the quantization range can recover exact
convergence~\cite{doan2021fast,reisizadeh2019exact}, achieving
linear rates for smooth strongly convex
objectives~\cite{kajiyama2021linear}. However, the communicated
object remains a $d$-dimensional vector, so the message length still
scales with the problem dimension. The second approach reduces the
number of transmitted entries. In distributed learning, QSGD showed
that aggressive quantization trades communication against estimator
variance~\cite{alistarh2017qsgd}, which motivated sparsification
methods that retain only a subset of entries and use error memory to
preserve discarded
information~\cite{stich2018memory,karimireddy2019error}. Koloskova
\emph{et al.}~\cite{koloskova2019choco} extend this idea to
decentralized networks, where agents maintain auxiliary estimates of
their neighbors' states to accommodate biased and unbiased
compressors. Compressed gradient-tracking and push--pull methods add
tracking variables to achieve linear convergence on directed
networks and nonconvex
objectives~\cite{liao2022compressed,liao2024robust}. The closest
schemes to our approach are random sparsification operators that
retain $k$ of the $d$
entries~\cite{stich2018memory,koloskova2019choco}. With $k=1$, the
rescaled operator $d\cdot\mathrm{rand}_1$ is equivalent to the
scaled coordinate-direction encoder considered in this paper. In
these methods, however, the selected indices are transmitted
together with the retained values~\cite{stich2018memory}, and the
resulting compression error is compensated using either error
memory~\cite{stich2018memory} or auxiliary neighbor
estimates~\cite{koloskova2019choco}. In summary, existing
communication-efficient methods either exchange messages whose
length scales with the problem dimension or rely on memory,
tracking, or auxiliary states to compensate for compression error.
In contrast, our scheme regenerates the communication direction, and
hence any selected index, from a shared seed, so each message
consists of a single real number and no compensating state is
required.

\emph{Contributions:} This paper introduces a scalar-communication mechanism for distributed optimization in which neighboring agents exchange a single real number per message, regardless of the decision dimension. The mechanism is developed for the continuous-time coordination algorithm of~\cite{SSK-JC-SM:15optim}, although it is applicable to any distributed optimization scheme whose inter-agent coupling depends only on relative states. By generating common random directions from a shared seed, agents communicate only scalar projections of their states, eliminating the need to transmit direction vectors or maintain compression-error memory. We show that the optimal consensus equilibrium of the original algorithm is preserved under scalar communication and remains the unique equilibrium on the zero-sum integral-state subspace. For strongly convex local objectives with globally Lipschitz gradients, we establish exponential mean-square and almost-sure convergence to the optimizer using constant gains, provided the direction-refresh interval is sufficiently small. The resulting refresh condition explicitly reveals how graph connectivity, objective regularity, and direction statistics influence convergence. A key challenge is that the communication direction is held fixed between refreshes. Although the resulting encoder is unbiased at each refresh instant, the held direction becomes coupled with the evolving network state over time. We address this challenge through a Lyapunov analysis that combines conditional drift at refresh times with bounds on the evolution between refreshes. Finally, we characterize the fresh-encoding variance of unnormalized Gaussian, Rademacher, scaled-coordinate, and spherical directions. The convergence theorem applies to the latter three fixed-norm distributions; scaled-coordinate directions recover random-coordinate communication as a special case. Compared to our prior work on scalar communication for federated learning~\cite{rostami2024fedscalar,rostami2026scalar}, this paper extends the framework from client-server aggregation to distributed neighbor-to-neighbor interactions, where the communicated quantities are dynamically coupled through the network.
% Optional, if proofs are omitted:
% For space reasons, proofs are omitted and will appear elsewhere.

\section{Problem Formulation and Baseline Algorithm}
\label{sec:problem}
\vspace{-0.05in}
We consider problem~\eqref{eq:global_opt} over a network of $N$
agents that communicate over a connected, undirected graph
$\mathcal{G}=(\mathcal{V},\mathcal{E})$ with symmetric adjacency
matrix $A=[a_{ij}]$, where $a_{ij}=a_{ji}>0$ if $(i,j)\in\mathcal{E}$
and $a_{ij}=0$ otherwise, and Laplacian~$L$.

\begin{assump}\label{assump:cost}
Each local cost $f^i:\mathbb{R}^d\to\mathbb{R}$ is private to
agent~$i$, continuously differentiable, and strongly convex, with a
globally Lipschitz continuous gradient. Let $m_f>0$ and $\ell_f$
denote a common lower bound on the strong convexity parameters and
a common upper bound on the gradient Lipschitz constants,
respectively.
\end{assump}

Under Assumption~\ref{assump:cost}, $F$ has a unique minimizer
$z^\star$. Several continuous-time algorithms solve
problem~\eqref{eq:global_opt} over connected graphs, including the
control-theoretic dynamics of Wang and
Elia~\cite{wang2010control,wang2011control} and their
variants~\cite{gharesifard2014distributed,SSK-JC-SM:15optim}.
These algorithms couple agents only through relative states
$x^i-x^j$, which is precisely the quantity our scalar encoding
targets. The encoding therefore applies to any of them, although
convergence must be established for each. For concreteness, we
develop and analyze it for the algorithm of Kia
\emph{et al.}~\cite{SSK-JC-SM:15optim}:
\begin{subequations}\label{eq:orig}
\begin{align}
  \dot{v}^i &= \alpha\beta \sum\nolimits_{j=1}^N a_{ij}\,(x^i - x^j),
  \label{eq:orig_v}\\
  \dot{x}^i &= -\alpha \nabla f^i(x^i)
               - \beta \sum\nolimits_{j=1}^N a_{ij}\,(x^i - x^j) - v^i,
  \label{eq:orig_x}
\end{align}
\end{subequations}
where $\alpha,\beta>0$ are design gains, $x^i\in\mathbb{R}^d$ is
agent~$i$'s estimate of $z^\star$, and $v^i\in\mathbb{R}^d$ is an
integral correction variable initialized so that
$\sum_{i=1}^N v^i(0)=0$. The following result establishes
convergence of this baseline.

\begin{thm}[See Theorem 6 of \cite{SSK-JC-SM:15optim}]\label{thm:kia}
% --- Paste the statement from Kia et al. here and match notation ---
Let Assumption~\ref{assump:cost} hold and let $\mathcal{G}$ be
connected and undirected. Then, for any $\alpha,\beta>0$ and any
initial condition satisfying $\sum_{i=1}^N v^i(0)=0$, the
trajectories of \eqref{eq:orig} converge
exponentially fast to $x^i=z^\star$ and
$v^i=-\alpha\nabla f^i(z^\star)$ for all $i\in\mathcal{V}$.
\end{thm}

%To simplify the exposition, we set $\alpha=\beta=1$ in the rest of the paper; the analysis extends directly to general $\alpha,\beta>0$. 
Implementing
\eqref{eq:orig} requires each agent to receive
the full $d$-dimensional state $x^j$ from every neighbor.
Section~\ref{sec:scalar} replaces this exchange with a single scalar
per neighbor, reducing the per-message cost from $\mathcal{O}(d)$
to $\mathcal{O}(1)$.

\section{Scalar Communication via Random Encoding and Decoding}
\label{sec:scalar}
\vspace{-0.05in}
The idea is to enforce consensus by letting neighbors act on a
rank-one surrogate of each other's state rather than on the state
itself. Given a random direction $u\in\mathbb{R}^d$ with
$\mathbb{E}[uu^\top]=I_d$, the surrogate
\begin{equation}\label{eq:surrogate}
    y^i = (u^\top x^i)\,u
\end{equation}
is an unbiased estimate of $x^i$, since
$\mathbb{E}[y^i]=\mathbb{E}[uu^\top]x^i=x^i$. Its usefulness for
communication lies in how it factors: if the receiver already knows
$u$, then $y^i$ is fully determined by the scalar $u^\top x^i$. In
the coupling terms of \eqref{eq:orig}, we
therefore replace $x^i-x^j$ by $y^i-y^j=uu^\top(x^i-x^j)$, which
requires each agent to send one real number per message instead of
$d$.

For this to work, all agents must hold the same direction at the
same time without exchanging it. We let $u_k\in\mathbb{R}^d$ be
drawn independently of the past at the refresh times $t_k=k\tau$,
$\tau>0$, and held constant on $[t_k,t_{k+1})$. A common seed, the
same pseudorandom generator, and a synchronized index~$k$ allow
every agent to regenerate $u_k$ locally, so the direction is never
transmitted. During $[t_k,t_{k+1})$, agent~$i$ transmits the
\emph{scalar encoding}
\begin{equation}\label{eq:encoding}
    c^i(t) = u_k^\top x^i(t) \in \mathbb{R},
\end{equation}
and each neighbor decodes $y^i(t)=c^i(t)\,u_k$. We stress that $y^i$
is not a reconstruction of $x^i$ from a single scalar; it is the
component of $x^i$ visible along $u_k$, and it is the refreshing of
$u_k$ that exposes all components over time. The resulting
algorithm is
\begin{subequations}\label{eq:mod}
\begin{align}
    \dot{v}^i &= \alpha\beta \sum\nolimits_{j=1}^N a_{ij} (y^i - y^j),
    \label{eq:mod_v} \\
    \dot{x}^i &= -\alpha \nabla f^i(x^i)
                 - \beta \sum\nolimits_{j=1}^N a_{ij} (y^i - y^j) - v^i.
    \label{eq:mod_x}
\end{align}
\end{subequations}
The local gradient is evaluated at the actual state $x^i$; only the
neighbor coupling is encoded and decoded. Each message thus carries
$\mathcal{O}(1)$ rather than $\mathcal{O}(d)$ entries. In this
continuous-time model, $c^i(t)$ varies continuously within each
holding interval; the finite-bit cost of an implementation also
depends on sampling and quantization, which we do not address here.

Any distribution with 
$\mathbb{E}[u_ku_k^\top]=I_d$
  can be used in the unbiased encoding identity \eqref{eq:surrogate}; no normalization is required for that identity. The convergence analysis in Section III-C additionally requires 
$\|u_k\|^2=d$ almost surely.  We consider
three representative choices:
\begin{enumerate}
  \item[(G)] \emph{Gaussian}: $g_k\sim\mathcal N(0,I_d)$, used either
directly, $u_k=g_k$, or normalized to the sphere,
$u_k=\sqrt d\,g_k/\|g_k\|$;
  \item[(R)] \emph{Rademacher}: $u_k\in\{-1,+1\}^d$ with
    independent, equiprobable entries;
  \item[(C)] \emph{scaled coordinate}: $u_k=\sqrt{d}\,e_{J_k}$ with
    $J_k$ uniform on $\{1,\ldots,d\}$, so that agents exchange one
    randomly selected coordinate at a time.
\end{enumerate}
Rademacher, scaled-coordinate, and sphere-normalized Gaussian
directions satisfy both isotropy and the fixed-norm condition
required by Theorem~\ref{thm:convergence}. Unnormalized Gaussian
directions satisfy isotropy but are not covered by
Theorem~\ref{thm:convergence}; they are included below only for the
fresh-encoding variance comparison. Among the covered choices,
Rademacher directions mix all $d$ coordinates into every transmitted
scalar, whereas scaled-coordinate directions expose one coordinate
per transmission; Section~\ref{sec:numerical} shows that either can
be preferable depending on where the disagreement between agents
lies.

\subsection{Equilibrium Analysis}
\label{subsec:equilibrium}
\vspace{-0.05in}
We first verify that the optimizer $z^\star$ of
\eqref{eq:global_opt} remains an equilibrium of
\eqref{eq:mod}. Let $(\bar{x},\bar{v})$ denote a
constant equilibrium of the state and integral variables; the
decoded variables $\bar y^i=u_ku_k^\top\bar x^i$ may still change
when the direction changes.

For any fixed direction, setting $\dot{v}^i=0$ at equilibrium in \eqref{eq:mod_v}~yields
\begin{equation}
    \sum\nolimits_{j=1}^N a_{ij}(\bar{y}^i-\bar{y}^j)=0
    \ \Longleftrightarrow\ (L\otimes I_d)\bar{y}=0 .
\end{equation}
Since $\mathcal{G}$ is connected, the null space of $L$ is spanned by
$\mathbf{1}_N$, so the decoded variables reach consensus. This alone
does not imply $\bar{x}^i=z^\star$, so we verify the desired
equilibrium directly. Set
\begin{equation}\label{eq:equilibrium}
    \bar{x}^i=z^\star,\qquad
    \bar{v}^i=-\alpha\nabla f^i(z^\star).
\end{equation}
For every shared direction $u_k$, the decoded states satisfy
$\bar{y}^1=\cdots=\bar{y}^N=u_ku_k^\top z^\star$, so all neighbor
differences vanish and
\begin{equation}\label{eq:proj}
    \dot{v}^i=0,\qquad
    \dot{x}^i=-\alpha\nabla f^i(z^\star)-\bar{v}^i=0 .
\end{equation}
Because the graph is undirected, the sum of the integral variables is
conserved,
\begin{equation}\label{eq:zero_sum}
    \sum\nolimits_{i=1}^N v^i(t)=\sum\nolimits_{i=1}^N v^i(0)=0,
\end{equation}
and the proposed equilibrium respects this condition since
$\sum_{i}\bar{v}^i=-\alpha\sum_{i}\nabla f^i(z^\star)=0$, which is
precisely the first-order optimality condition for $F$. Thus the
optimal consensus state is an exact equilibrium for every shared
direction.

Conversely, if $\mathbb{E}[u_ku_k^\top]=I_d$, any constant
equilibrium common to almost every direction satisfies
$(L\otimes u_ku_k^\top)\bar x=0$ for almost every $u_k$; averaging
gives $(L\otimes I_d)\bar{x}=0$, so $\bar x^i$ is common to all
agents. Summing the equilibrium equations for $\dot{x}^i$ then gives
$\sum_i\nabla f^i(\bar{x}^i)=0$. Hence, on the zero-sum subspace
\eqref{eq:zero_sum}, the common equilibrium
\eqref{eq:equilibrium} is unique.

\subsection{Why a Fixed Direction Can Fail}
\label{subsec:fixed}
\vspace{-0.05in}
Although $z^\star$ is an equilibrium for every direction, holding a
single direction $u$ fixed for all time need not yield convergence
when $d>1$. Let $P=uu^\top$ denote the rank-one encoding--decoding
operator (an orthogonal projection only when $\|u\|=1$). The
coupling depends on the states only~through
\begin{equation}
    y^i-y^j=P(x^i-x^j),
\end{equation}
so if $x^i-x^j\in\operatorname{span}(u)^\perp$ for adjacent agents,
then $y^i-y^j=0$ and the coupling is blind to disagreement in the
remaining $d-1$ dimensions. A state profile is therefore a spurious
equilibrium whenever
\begin{align}
    x^i-x^j&\in\operatorname{span}(u)^\perp,
       &&\forall(i,j)\in\mathcal E,\\
    v^i&=-\alpha\nabla f^i(x^i),
       &\sum\nolimits_{i=1}^N\nabla f^i(x^i)&=0 .
\end{align}
For example, take two agents with $f^1(z)=\frac12\|z-a\|^2$ and
$f^2(z)=\frac12\|z+a\|^2$, where $a\ne0$ and $u^\top a=0$. The
optimizer is $z^\star=0$, yet $x^1=a$, $x^2=-a$, $v^1=v^2=0$ is a
nonconsensus equilibrium. A fixed direction thus offers no general
convergence guarantee, which motivates refreshing the direction.

% ------------------------------------------------------------------
% ------------------------------------------------------------------
% Revised Section III-C.  Uses the paper's labels (eq:mod, eq:mod_v,
% eq:mod_x, assump:cost, eq:global_opt) and environments thm, rem;
% a lemma environment `lem' is assumed (\newtheorem{lem}{Lemma}).
% ------------------------------------------------------------------
\subsection{Convergence under Direction Refreshing}
\label{subsec:convergence}
\vspace{-0.05in}
Let $u_k$ be redrawn independently from a fixed distribution at the
refresh times $t_k=k\tau$. Unbiasedness alone does not guarantee
convergence for arbitrary holding intervals, so we derive a positive
constant $\bar\tau_d$, determined by the problem data alone, such
that convergence is preserved for every holding interval
$\tau\in(0,\bar\tau_d]$. Any refresh period in this range is
admissible, so $\tau$ can be chosen as large as $\bar\tau_d$ to
refresh directions as infrequently as the guarantee allows. Since
the refresh times are periodic rather than state-triggered, no
Zeno-type issue arises.

Let $\mathcal F_k$ denote the history immediately before drawing
$u_k$, and assume
\begin{equation}\label{eq:isotropy}
    \mathbb E[u_ku_k^\top\mid\mathcal F_k]=I_d,\qquad
    \|u_k\|^2=d .
\end{equation}
The first condition holds for (G), (R), and (C). The second holds for
(R), (C), and normalized Gaussian directions, which remain isotropic
(Remark~\ref{rem:gauss}); it fails for unnormalized Gaussian
directions, whose squared norm $\|g_k\|^2$ is random. At a
refresh time the states are $\mathcal F_k$-measurable, so
\begin{align}
 \mathbb E[y^i(t_k^+)-y^j(t_k^+)\mid\mathcal F_k]
 &=\mathbb E[u_ku_k^\top\mid\mathcal F_k]\,(x^i(t_k)-x^j(t_k))\nonumber\\
 &=x^i(t_k)-x^j(t_k).
\end{align}
Within a holding interval, however, the states depend on the sampled
direction, so this factorization is available only at refresh times.
The analysis below controls this dependence through the component of
the state along $u_k$.

\emph{Error coordinates.}
Let $x=\operatorname{col}(x^1,\ldots,x^N)$ and define the stacked cost
$G(x)=\sum_{i=1}^N f^i(x^i)$, whose gradient stacks the local
gradients (in contrast to the common-argument cost $F$ in
\eqref{eq:global_opt}). Let
\begin{equation}
    x^\star=\mathbf1_N\otimes z^\star,\qquad
    v^\star=-\alpha\nabla G(x^\star),
\end{equation}
and define the errors $\tilde x=x-x^\star$, $\tilde v=v-v^\star$ and
$h=\nabla G(x)-\nabla G(x^\star)$. Following~\cite{SSK-JC-SM:15optim},
let $r=\mathbf1_N/\sqrt N$ and let $R\in\mathbb R^{N\times(N-1)}$
satisfy $r^\top R=0$, $R^\top R=I_{N-1}$, and
$RR^\top=I_N-\tfrac1N\mathbf1_N\mathbf1_N^\top$, and apply the
orthogonal change of variables
\begin{equation}\label{eq:cov}
 z=\big([r\ R]^\top\!\otimes I_d\big)\tilde x,\qquad
 w=\big([r\ R]^\top\!\otimes I_d\big)\tilde v,
\end{equation}
partitioned as $z=(z_1,z_2)$ and $w=(w_1,w_2)$ with
$z_1,w_1\in\mathbb R^d$. Here $z_1$ is the (scaled) average error and
$z_2$ the disagreement. Let
$\hat L=R^\top LR\in\mathbb R^{(N-1)\times(N-1)}$, which is symmetric
positive definite with eigenvalues $\lambda_2\le\cdots\le\lambda_N$
(the positive eigenvalues of $L$), and satisfies $LR=R\hat L$. Write
$\bar h=([r\ R]^\top\otimes I_d)h=(h_1,h_2)$, so that
$z^\top\bar h=\tilde x^\top h$ and $\|\bar h\|=\|h\|$.
Because $r^\top L=0$, along \eqref{eq:mod_v}--\eqref{eq:mod_x} we have
$\dot w_1=0$, and the initialization $\sum_iv^i(0)=0$ gives
$w_1(t)\equiv0$: this is the constant of motion identified
in~\cite{SSK-JC-SM:15optim}. The remaining error dynamics, with
$P=u_ku_k^\top$ on $[t_k,t_{k+1})$, are
\begin{subequations}\label{eq:err_dyn}
\begin{align}
 \dot z_1&=-\alpha h_1,\label{eq:z1}\\
 \dot z_2&=-\alpha h_2-\beta(\hat L\otimes P)z_2-w_2,\label{eq:z2}\\
 \dot w_2&=\alpha\beta(\hat L\otimes P)z_2 .\label{eq:w2}
\end{align}
\end{subequations}
The reduced state $\zeta=\operatorname{col}(z_1,z_2,w_2)$ has the
origin as its unique equilibrium common to almost every direction
realization on the zero-sum integral-state subspace; for an
individual fixed direction, additional spurious equilibria may
exist, as shown in Section~III-B. Since \eqref{eq:cov} is an
isometry, $\|\zeta\|^2=\|\tilde x\|^2+\|\tilde v\|^2$, and
convergence of $\zeta$ is equivalent to convergence of $(x,v)$ to
$(x^\star,v^\star)$. The zero eigenvalue of $L$ has been removed: no
invariant subspace needs to be tracked, and all Lyapunov functions
below are positive definite on the whole reduced state space.
\emph{Projected states.} On $[t_k,t_{k+1})$ let $\hat u=u_k/\sqrt d$,
so $P=d\,\hat u\hat u^\top$, and for $y\in\mathbb R^{(N-1)d}$ let
$y^\parallel=(I_{N-1}\otimes\hat u^\top)y\in\mathbb R^{N-1}$ collect
the components along $u_k$. Write $a=z_2^\parallel$,
$b=w_2^\parallel$, and $h^\parallel=h_2^\parallel$. Then
$(\hat L\otimes P)y=d\,(\hat Ly^\parallel)\otimes\hat u$ and, for
$\mathcal F_k$-measurable $y$,
\begin{equation}\label{eq:proj_moment}
 \|y^\parallel\|\le\|y\|,\qquad
 \mathbb E\big[\|y^\parallel\|^2\mid\mathcal F_k\big]=\|y\|^2/d .
\end{equation}
The projected dynamics of \eqref{eq:z2}--\eqref{eq:w2} are
\begin{equation}\label{eq:proj_dyn}
 \dot a=-\alpha h^\parallel-\beta d\,\hat La-b,\qquad
 \dot b=\alpha\beta d\,\hat La .
\end{equation}
The only $O(d)$ term in \eqref{eq:err_dyn}, $-\beta d(\hat La)\otimes\hat u$,
is dissipative. It acts only on $a$, whose mean square at a refresh
is $\|z_2\|^2/d$ by \eqref{eq:proj_moment}. The average error $z_1$
is not affected by the direction at all.

\emph{Lyapunov function.} Define
\begin{align}
 b_-&=\min\{1,1/(\alpha\beta\lambda_N)\},\qquad
 b_+=\max\{1,1/(\alpha\beta\lambda_2)\},\nonumber\\
 B&=\alpha\ell_f+\beta\lambda_N ,
\end{align}
choose
\begin{equation}\label{eq:epsilon_choice}
 0<\varepsilon\le\min\Big\{\frac{b_-}{2},\ \frac{\alpha m_f}{B^2},\
 \frac1{2\alpha}\Big\},
\end{equation}
and set $c_-=b_--\varepsilon$, $c_+=b_++\varepsilon$,
$c_0=\tfrac12\min\{\alpha m_f,\varepsilon\}$, and
$\kappa=1+\varepsilon\beta\lambda_N$. Consider
\begin{equation}\label{eq:V}
 V_\varepsilon=\frac12z^\top z
    +\frac1{2\alpha\beta}w_2^\top(\hat L^{-1}\otimes I_d)w_2
    +\varepsilon z_2^\top w_2 .
\end{equation}
Since $\lambda_N^{-1}I\preceq\hat L^{-1}\preceq\lambda_2^{-1}I$ and
$|\varepsilon z_2^\top w_2|\le\tfrac\varepsilon2\|\zeta\|^2$,
\begin{equation}\label{eq:V_bounds}
 \frac{c_-}{2}\|\zeta\|^2
 \leq V_\varepsilon(\zeta)
 \leq\frac{c_+}{2}\|\zeta\|^2 ,
\end{equation}
so $V_\varepsilon$ is positive definite and radially unbounded on the
reduced state space. Using this Lyapunov function, we establish in
Appendix three auxiliary results that lead to our
main theorem. Lemma~\ref{lem:growth} bounds the growth of the state
over one holding interval by a constant independent of $d$;
Lemma~\ref{lem:decomp} decomposes $\dot V_\varepsilon$ along
\eqref{eq:err_dyn} into a strictly dissipative part, a term absorbed
by the $O(d)$ dissipation in \eqref{eq:proj_dyn}, and a term that has
zero conditional mean at refresh times; and
Lemma~\ref{lem:moments} shows that the projected quantities
$a$, $b$, $h^\parallel$, and the deviation $\zeta(t)-\zeta(t_k)$,
have conditional mean square $O(1/d)$ over the holding interval.
Together these control the dependence of the held direction on the
evolving state and yield the following result.

\begin{thm}\label{thm:convergence}
Let Assumption~\ref{assump:cost} hold, let $\mathcal G$ be connected
and undirected, let $\sum_iv^i(0)=0$, and let the directions $u_k$ be
drawn independently across refreshes so that \eqref{eq:isotropy}
holds. Let $\bar\gamma=\gamma(1/\bar C)$ and
$\bar\theta=1+\alpha^2\ell_f^2(1+C_h^2\bar\gamma^2/4)$, and define
the $d$-independent constants
\begin{align*}
 \rho_0&=\min\Big\{\frac1{4\alpha},
 \frac1{4\sqrt2\,\alpha\beta\lambda_N},\\
 &\qquad\frac1\kappa\sqrt{\frac{c_0}{128\beta\lambda_N}},
 \frac1{32\beta\lambda_N\sqrt{2+\alpha^2}}\Big\},\\
 \sigma_0&=\min\Big\{1,\frac1{\alpha\beta\lambda_N(1+\ell_f)},
 \frac1{64C_h^2\bar\gamma^2}\Big\},\\
 \tau_0&=\min\Big\{1,\frac1{\bar C},\frac1{\sqrt{\bar\theta}},
 \frac{c_0}{16K_0}\Big\},
\end{align*}
where
$K_0=2\sqrt3\,\alpha\kappa\ell_f+8\sqrt2\,\beta\kappa\lambda_N
+4\varepsilon\beta\lambda_N+32\varepsilon\alpha\beta^2\lambda_N^2
+c_7\bar\gamma^2$, and let
\begin{equation}\label{eq:tau_bar}
 \bar\tau_d=\min\Big\{\tau_0,\ \frac{\rho_0}{d},\
 \sqrt{\frac{\sigma_0}{d}}\Big\}>0 .
\end{equation}
If $\tau\in(0,\bar\tau_d]$, then the trajectories of
\eqref{eq:mod_v}--\eqref{eq:mod_x} converge to $(x^\star,v^\star)$
exponentially fast in mean square and almost surely. In particular,
with $r=1-c_0\tau/(2c_+)$,
\begin{equation}\label{eq:mean_square_rate}
 \mathbb E\|\zeta(t_k)\|^2
 \leq\frac{c_+}{c_-}\,r^k\,\|\zeta(0)\|^2 .
\end{equation}
\end{thm}

\begin{proof}
The condition $\tau\le\bar\tau_d$ is equivalent to
\begin{equation}\label{eq:refresh_condition}
 \tau\le\tau_0,\qquad d\,\tau\le\rho_0,\qquad d\,\tau^2\le\sigma_0,
\end{equation}
and $\bar\tau_d$ depends only on the problem data, so refresh times are
separated by the fixed interval $\tau>0$ and every interval $[0,T]$
contains finitely many of them; on each holding interval the
right-hand side of \eqref{eq:err_dyn} is globally Lipschitz, so the
solution exists, is unique, and by Lemma~\ref{lem:growth} stays
bounded. Integrate \eqref{eq:Vdot_decomp} over the holding interval
and write
$\int_0^\tau S_u\,ds=\tau S_u(\zeta_k)+\int_0^\tau(\tau-s)\dot S_u\,ds$.
This gives
\begin{align*}
 &V_\varepsilon(t_{k+1})-V_\varepsilon(t_k)\le
 -c_0\!\int_0^\tau\!\|\zeta\|^2ds\\
 &\quad-\beta d(1-\varepsilon\alpha)\!\int_0^\tau\!a^\top\hat La\,ds
 +\tau S_u(\zeta_k)+\tau\!\int_0^\tau\![\dot S_u]_+\,ds .
\end{align*}
Bound $[\dot S_u]_+$ using \eqref{eq:Sdot}. Since
$1-\varepsilon\alpha\ge\tfrac12$ and $\alpha\rho\le\tfrac14$, the first
term of \eqref{eq:Sdot}, multiplied by $\tau$, is at most
$\tfrac12\beta d(1-\varepsilon\alpha)a^\top\hat La$. By Young's
inequality, $\tau$ times the second term is at most
$\tfrac14\beta d(1-\varepsilon\alpha)a^\top\hat La
+2\beta\kappa^2\lambda_Nd^3\tau^2\|b\|^2$. The dissipative integral
absorbs both $a^\top\hat La$ parts along every trajectory.

Next, take conditional expectations. The term $\tau S_u(\zeta_k)$
vanishes by Lemma~\ref{lem:decomp}. Condition
\eqref{eq:refresh_condition} ensures $\kappa_a\le32$, $\kappa_b\le4$,
$\kappa_h\le3\ell_f^2$, and $\kappa_\delta\le\tfrac18$. Lemma~\ref{lem:moments}
and the Cauchy--Schwarz inequality then give:
\begin{itemize}
\item the $\|b\|^2$ remainder is at most
$2\beta\kappa^2\lambda_N\kappa_b\rho^2\tau\|\zeta_k\|^2
\le\tfrac{c_0}{16}\tau\|\zeta_k\|^2$;
\item the remaining terms of \eqref{eq:Sdot} contribute at most
$\tau^2K_0\|\zeta_k\|^2\le\tfrac{c_0}{16}\tau\|\zeta_k\|^2$;
\item since $\|\zeta(s)\|^2\ge\frac12\|\zeta_k\|^2-\|\zeta(s)-\zeta_k\|^2$,
the first integral contributes at most
$-\tfrac38c_0\tau\|\zeta_k\|^2$ by \eqref{eq:zeta_dev}.
\end{itemize}
Hence
\[
 \mathbb E[V_\varepsilon(t_{k+1})\mid\mathcal F_k]-V_\varepsilon(t_k)
 \le-\tfrac{c_0\tau}{4}\|\zeta_k\|^2
 \le-\tfrac{c_0\tau}{2c_+}V_\varepsilon(t_k),
\]
and iterating with \eqref{eq:V_bounds} gives \eqref{eq:mean_square_rate}.
Between refreshes, Lemma~\ref{lem:growth} gives
$\mathbb E[\sup_{t_k\le t\le t_{k+1}}\|\zeta(t)\|^2]
\le\bar\gamma^2\mathbb E\|\zeta(t_k)\|^2$. For any
$\bar r\in(r,1)$, Markov's inequality shows that
$\sum_k\mathbb P(\sup_{[t_k,t_{k+1}]}\|\zeta\|^2>\bar r^k)<\infty$, and
the Borel--Cantelli lemma yields almost sure exponential convergence.
Since \eqref{eq:cov} is an isometry, the same statements hold for
$(\tilde x,\tilde v)$.
\end{proof}

\begin{rem}[Dimension dependence]\label{rem:dimension}
For $d\ge\rho_0^2/\sigma_0$, \eqref{eq:tau_bar} reduces to
$\bar\tau_d=\min\{\tau_0,\rho_0/d\}$, so the admissible holding
interval scales as $1/d$, while the decay rate per unit time,
$c_0/(2c_+)$, does not depend on $d$ and matches, up to constants,
the rate the same Lyapunov argument gives for full-vector exchange.
A bound based only on the worst-case norm
$\|\hat L\otimes u_ku_k^\top\|=\lambda_Nd$ would give
$\tau=O(1/d^2)$; the improvement comes from the two facts noted
after \eqref{eq:proj_dyn}. In the continuous-time model, refreshing
costs no communication, since $u_k$ is regenerated locally, and each
link carries one real-valued signal instead of $d$ for every
admissible $\tau$. In a sampled implementation that transmits one
scalar per refresh, however, the condition $\tau\le\bar\tau_d=O(1/d)$
implies $O(d)$ scalar transmissions per unit time per link, the same
order as one full-vector message per unit time. The guarantee
therefore establishes a reduction in message size, which is what a
per-message bandwidth limit constrains, and not a reduction in total
scalar throughput. The example in Section~\ref{sec:numerical}
converges at refresh intervals well above $\bar\tau_d$, which
suggests the bound is conservative; whether the throughput also
decreases is an open question that requires the sampled-data
analysis left to future work.
\end{rem}

%\begin{rem}[Conservativeness]\label{rem:conservative}
%The constants in \eqref{eq:tau_bar} are sufficient, not necessary, but
%the $1/d$ scaling and the need for some bound are both visible in the
%example of Section~IV: for $d\in\{16,64,256\}$ the decay rate is
%retained for $d\tau\lesssim0.4$, degrades as $d\tau$ grows, and the
%trajectories diverge once $d\tau$ exceeds a threshold between $3$ and
%$6$ that is essentially independent of $d$. The theorem does not
%assert stability for arbitrary refresh intervals, and the example
%shows that none should be expected. The threshold is governed by the
%term $1/(\alpha\beta\lambda_N)$ in $\rho_0$, which bounds the integral
%action $\alpha\beta d\hat L$ accumulated over one holding interval.
%\hline
%The constants in (35) are sufficient, not necessary. In the example of Section IV, with $d=64$, the error at a fixed time grows with $d\tau$ well beyond $\bar\tau_d$, as Fig. 2(b) shows, and the threshold is governed by the term $1/(\alpha\beta\lambda_N)$ in $\rho_0$, which bounds the integral action accumulated over one holding interval.

\begin{rem}[Conservativeness]
The constants in \eqref{eq:tau_bar} are sufficient, not necessary.
In the example of Section~\ref{sec:numerical}, convergence is
observed for holding intervals well beyond $\bar\tau_d$, and the
error at a fixed time grows with $d\tau$ (Fig.~\ref{fig:refresh}b),
so some bound is needed even though \eqref{eq:tau_bar} is not tight.
The threshold is governed by the term $1/(\alpha\beta\lambda_N)$ in
$\rho_0$, which bounds the integral action $\alpha\beta d\hat L$
accumulated over one holding interval.
\end{rem}

\begin{rem}[Gaussian directions]\label{rem:gauss}
Lemmas~\ref{lem:decomp} and~\ref{lem:moments} use $\|u_k\|^2=d$
through $P=d\,\hat u\hat u^\top$; for $u_k=g_k$ the factor $d$ is
replaced by $\|g_k\|^2$, which is unbounded, which is why no refresh
bound of the form \eqref{eq:tau_bar} follows. Normalizing to the
sphere of radius $\sqrt d$ restores \eqref{eq:isotropy} and gives
encoding variance $(d-1)\|a\|^2$, the same as (R) and (C).
\end{rem}
\subsection{Encoding Variance and Exactness of Convergence}
\label{subsec:variance}
\vspace{-0.05in}
Theorem~\ref{thm:convergence} guarantees convergence to $z^\star$
under the refresh condition. Random encoding nevertheless introduces
transient fluctuations, which we now quantify at the refresh~times.

Let $P(t)=I_N\otimes(u_ku_k^\top)$ on $[t_k,t_{k+1})$ and decompose
\begin{equation}
 P(t)=I_{Nd}+\boldsymbol\Delta_k,\qquad
 \mathbb E[\boldsymbol\Delta_k\mid\mathcal F_k]=0 .
\end{equation}
With the coupling perturbation
$\xi(t)=(L\otimes I_d)\boldsymbol\Delta_k\tilde x(t)$, the error
dynamics read
\begin{align}\label{eq:error_perturbation}
 \dot{\tilde x}
    &=-\alpha h-\beta(L\otimes I_d)\tilde x-\tilde v-\beta\xi(t),\\
 \dot{\tilde v}
    &=\alpha\beta(L\otimes I_d)\tilde x+\alpha\beta\xi(t).\nonumber
\end{align}
This is a random differential equation with piecewise-constant
directions, not an It\^{o} stochastic differential equation; a
diffusion approximation would require a separate limiting argument.
The perturbation enters both the state and the integral equations,
is multiplicative in the disagreement, and vanishes at consensus.

\subsubsection*{Encoding variance}
For a fixed $a\in\mathbb R^d$, unnormalized Gaussian fourth moments
give $\mathbb E[(uu^\top)^2]=(d+2)I_d$, hence
\begin{equation}\label{eq:gauss_variance}
 \mathbb E\|(uu^\top-I_d)a\|^2=(d+1)\|a\|^2 .
\end{equation}
For any fixed-norm direction, including (R), (C), and
sphere-normalized Gaussian, $\|u\|^2=d$ and
$\mathbb E[(uu^\top)^2]=dI_d$, so
\begin{equation}\label{eq:bounded_variance}
 \mathbb E\|(uu^\top-I_d)a\|^2=(d-1)\|a\|^2 .
\end{equation}
Consequently, at each refresh time,
\begin{align}
 \mathbb E[\xi(t_k^+)\mid\mathcal F_k]&=0,\\
 \mathbb E[\|\xi(t_k^+)\|^2\mid\mathcal F_k]
   &=\omega\|(L\otimes I_d)\tilde x_k\|^2
   \leq\omega\lambda_N^2\|\tilde x_k\|^2,
\end{align}
with $\omega=d+1$ for unnormalized Gaussian and $\omega=d-1$ for
fixed-norm directions. Fixed-norm directions thus reduce the
encoding variance by $2\|(L\otimes I_d)\tilde x_k\|^2$ relative to
the unnormalized Gaussian choice, the direct analogue of the variance
reduction established for federated aggregation
in~\cite{rostami2024fedscalar}. The fixed-norm choices differ in how
this variance is distributed: coordinate directions concentrate it on
the sampled coordinate, whereas Rademacher directions spread it over
all coordinates and yield a transmitted scalar that is sub-Gaussian
with parameter $\|x^i\|$ rather than of magnitude up to
$\sqrt d\,\|x^i\|$. For all of these unbiased encoders the variance
grows linearly with $d$; it is the refresh condition, not the choice
of encoder, that controls the error. These identities concern a
freshly drawn direction; no independence is assumed within a holding
interval.

\subsubsection*{Exact convergence}
Let $X(t)=\mathbb E[\tilde x(t)\tilde x(t)^\top]$. By
Theorem~\ref{thm:convergence},
$\lim_{t\to\infty}\operatorname{tr}X(t)
=\lim_{t\to\infty}\mathbb E\|\tilde x(t)\|^2=0$. Hence, with constant
positive gains and sufficiently frequent refreshes, scalar
communication yields exact convergence: random encoding does not
leave a residual error ball. This is consistent with
\eqref{eq:error_perturbation}, since $\xi$ vanishes when all agents
agree. In particular, no diminishing coupling gain is required. The
result concerns the held-direction continuous-time model; a
sampled-message or discrete-time implementation requires its own
step-size and sampling analysis.
\section{Numerical Example}
\label{sec:numerical}
\vspace{-0.05in}
We consider distributed $\ell_2$-regularized logistic regression
over $N=20$ agents with decision dimension $d=64$. Agent~$i$
minimizes the local cost
\begin{equation}
 f^i(x)=\frac{1}{m}\sum\nolimits_{s=1}^{m}
 \log\!\left(1+e^{-y_{i,s}a_{i,s}^{\top}x}\right)
 +\frac{0.3}{2}\|x\|^2,
\end{equation}
where $a_{i,s}\in\mathbb R^{64}$ and $y_{i,s}\in\{-1,+1\}$.
Each agent communicates with the agents one and three positions
away in either direction around a cycle, with edge weights $1/8$.
We set $\alpha=\beta=1$ and initialize $x^i(0)=v^i(0)=0$.

\begin{figure}[t]
 \centering
 \includegraphics[width=\columnwidth]{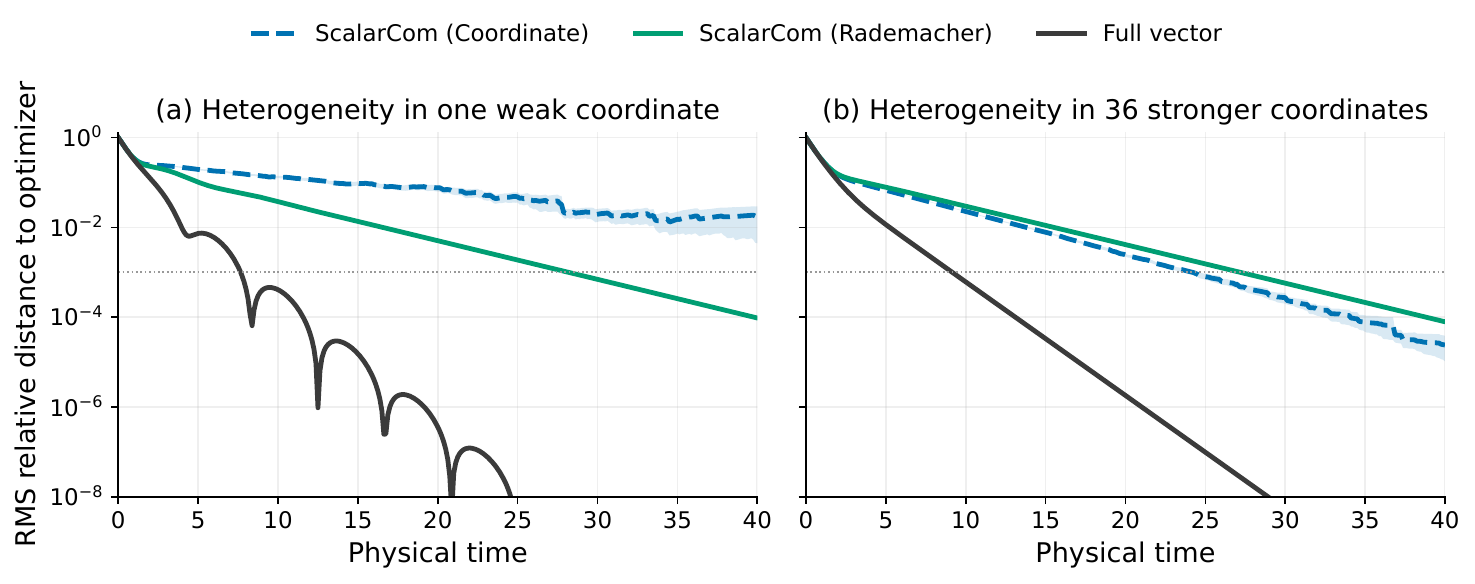}
 \caption{Relative RMS error at $d\tau=2$ for concentrated
 (left) and spread (right) gradient differences. Results use
 4096 runs; shaded bands show pointwise 95\% bootstrap
 confidence intervals for the RMS estimate.}
 \label{fig:geometry}\vspace{-0.1in}
\end{figure}

% \begin{figure}[t]
%  \centering
%  \includegraphics[width=\columnwidth]{Fig/fig_refresh.pdf}
%  \caption{Effect of refreshing in the concentrated case:
%  (a) error versus time; (b) RMS error at $t=40$ versus $d\tau$.
%  Error bars show pointwise 95\% bootstrap confidence intervals.}
%  \label{fig:refresh}
% \end{figure}
We construct synthetic data to control which coordinates cause
disagreement between agents. Each agent receives 128 samples with
feature vector $\gamma_j e_j$ for each coordinate $j$, giving
$m=8192$ samples in total. Here, $e_j$ is the $j$th coordinate
vector, $\gamma_j=16$ for $j\le49$, and $\gamma_j=2$ otherwise.
Thus, each sample involves a single coordinate, and the feature
magnitudes create different curvatures across coordinates.
The local costs are strongly convex and have globally Lipschitz
gradients, with $m_f=0.3$ and $\ell_f\le1.3$.

Agents are divided into two groups that alternate around the cycle.
The groups have identical feature vectors but different labels.
For coordinate~1, every agent has 96 positive and 32 negative
labels. We then consider two cases. In the \emph{concentrated}
case, the two groups have 16 and 112 positive labels, respectively,
among the 128 samples for coordinate~50; every other coordinate
except coordinate~1 has equally many positive and negative labels.
In the \emph{spread} case, the groups instead have 63 and 65
positive labels for each coordinate from 2 through 37; all
coordinates other than these and coordinate~1 are balanced.
The differences between agents therefore occur in one
weak-curvature coordinate in the first case and in 36
stronger-curvature coordinates in the second. Both cases have
the same aggregate objective and optimizer $z^\star$.
They also have the same magnitude of local gradient differences:
$\|\nabla f^i(x)-N^{-1}\sum_j\nabla f^j(x)\|=3/256$.
This construction lets us compare the effect of where those
differences occur while keeping their overall magnitude~fixed.

We compare scalar communication with shared Rademacher directions
$u_k\in\{-1,+1\}^{64}$, scalar communication with shared coordinate
directions $u_k=\sqrt{64}\,e_{J_k}$ for uniformly sampled $J_k$,
and full-vector exchange. Directions are drawn independently at
intervals of length $\tau$ and held constant between refreshes.
We simulate the continuous-time dynamics using fourth-order
Runge--Kutta integration, with steps ending at each refresh.
The integration step is at most
$\min\{\tau/4,\,0.4/(d+1.3)\}$ for all methods.
For $M$ independent runs, the relative RMS error is
\begin{equation}
 E_{\mathrm{RMS}}(t)=
 \frac{1}{\|z^\star\|}
 \left(\frac{1}{MN}\sum_{r=1}^{M}\sum_{i=1}^{N}
 \|x_r^i(t)-z^\star\|^2\right)^{1/2}.
\end{equation}
This measures the distance to the optimizer, averaged over agents
and runs, relative to its initial value. Since all agents start
at zero, $E_{\mathrm{RMS}}(0)=1$.

Figure~\ref{fig:geometry} compares the two data sets at
$\tau=2/64$, using $M=4096$ runs per stochastic method.
In the concentrated case, Rademacher encoding gives a substantially
smaller RMS error: at $t=20$, its error is $0.00504$, compared with
$0.0760$ for coordinate encoding. The ordering reverses in the
spread case, where coordinate encoding reaches $0.00245$,
compared with $0.00414$ for Rademacher encoding.
In this separable problem, coordinate feedback acts on individual
coordinates, whereas Rademacher feedback couples coordinates with
different curvatures. The comparison shows that this mixing can
help or hurt depending on where the differences between agents
occur. These are RMS comparisons: coordinate encoding can have
a smaller median error but a larger RMS error because some runs
converge much more slowly. Thus, neither direction distribution
is uniformly better, consistent with the theorem providing
convergence guarantees for both.

\begin{figure}[t]
 \centering
 \includegraphics[width=\columnwidth]{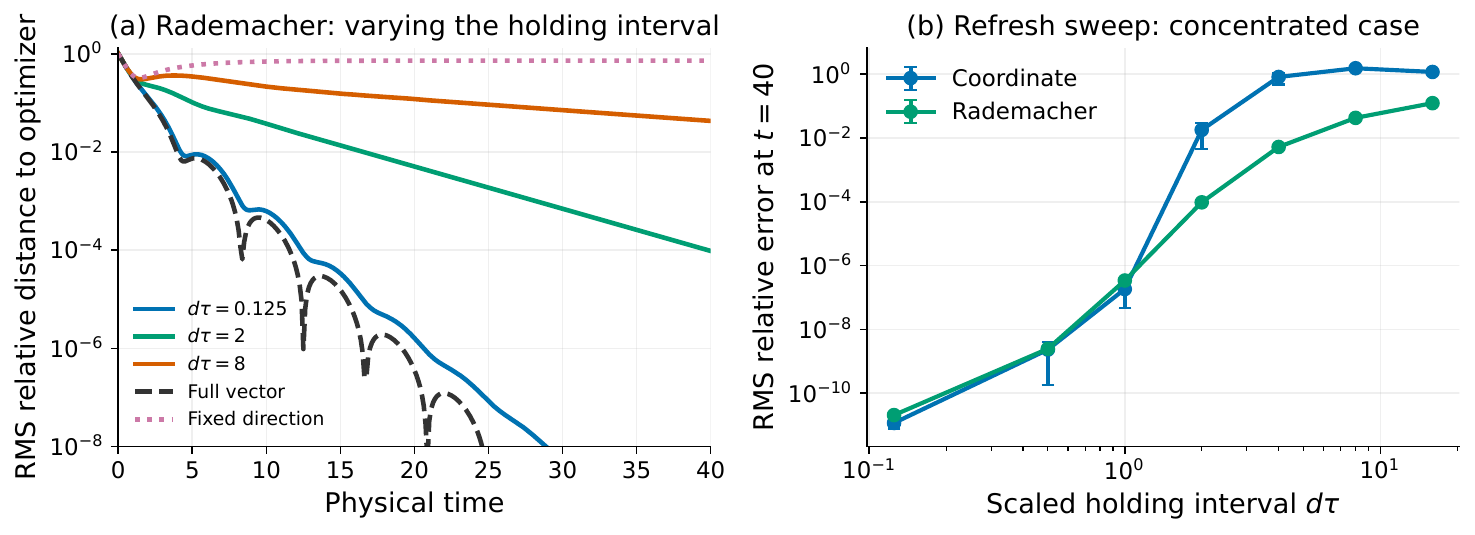}
 \caption{Effect of refreshing in the concentrated case:
 (a) error versus time; (b) RMS error at $t=40$ versus $d\tau$.
 Error bars show pointwise 95\% bootstrap confidence intervals.}
 \label{fig:refresh}
\end{figure}

Figure~\ref{fig:refresh} shows the effect of the holding interval
in the concentrated case. Panel~(a) shows that frequent refreshing
leads to faster error decay, whereas a fixed direction leaves a
persistent error. Panel~(b) reports the RMS error at $t=40$ for
different values of $d\tau$. Smaller holding intervals give lower
errors for both encoders, while Rademacher encoding performs
better at larger holding intervals in this example. 

These results illustrate how scalar communication can support
convergence to the optimizer while reducing the size of the
exchanged information. Each agent transmits only the scalar
$u_k^\top x^i$, and its neighbors reconstruct the shared direction
locally. The communicated state therefore has one component
instead of $d$, reducing the per-message payload from
$\mathcal{O}(d)$ to $\mathcal{O}(1)$ in a sampled implementation.
This makes the mechanism
attractive for channels that can accommodate a scalar message
but cannot carry a full state vector within the available
transmission budget. The continuous-time results demonstrate
convergence with scalar-valued exchange, providing the basis
for such bandwidth-constrained implementations.

\section{Conclusions}
\label{sec:conclusion}
\vspace{-0.05in}
We introduced a scalar-communication mechanism for consensus-based
distributed optimization in which each neighbor message carries a
single real number regardless of the decision dimension. Agents
regenerate a common random direction from a shared seed, transmit
the projection of their state onto that direction, and act on a
rank-one surrogate of their neighbors' states. We showed that the
optimal consensus equilibrium is preserved, that fixed directions
can introduce spurious equilibria, and that sufficiently frequent
direction refreshes yield exponential mean-square and almost-sure
convergence with constant gains and no residual error for strongly
convex objectives with Lipschitz gradients. Simulations illustrated
the effects of the direction distribution and refresh interval.
Future work includes a discrete-time formulation with sampled and
quantized scalar messages, tighter refresh-interval bounds,
extensions to directed and time-varying networks, robustness to
unreliable communication channels through redundant scalar
transmissions, and the combination of scalar communication with
event-triggered scheduling to reduce both message size and
communication~frequency.
\bibliographystyle{ieeetr}%
\bibliography{bib/alias,bib/Reference,bib/ref_check} % 

%\appendices
%\section{Auxiliary Results}
%\label{app:aux}
\vspace{-0.1in}
\appendix[Auxiliary Results]
\setcounter{equation}{0}
\renewcommand{\theequation}{A.\arabic{equation}}
\setcounter{lem}{0}
\renewcommand{\thelem}{A.\arabic{lem}}
\vspace{-0.05in}
\begin{lem}[Dimension-free growth]\label{lem:growth}
For any direction and $0\le s\le\tau$,
$\|\zeta(t_k+s)\|\le\gamma(s)\|\zeta(t_k)\|$, where
$\gamma(s)=(1+\alpha)\sqrt{1+(1+\alpha)^2}\,e^{\bar Cs}$ and
$\bar C=\max\{1+\alpha,\ \alpha+\alpha^2(1+\ell_f)\}$.
\end{lem}
\begin{proof}
Let $q=w_2+\alpha z_2$. By \eqref{eq:z2}--\eqref{eq:w2},
$\dot q=-\alpha(\alpha h_2+w_2)$, so the direction-dependent terms
cancel. Also, by \eqref{eq:err_dyn} and $z^\top\bar h\ge0$,
$\tfrac{d}{dt}\tfrac12\|z\|^2=-\alpha z^\top\bar h
-\beta d\,a^\top\hat La-z_2^\top w_2\le\|z\|\|w_2\|$.
Using $\|w_2\|\le\|q\|+\alpha\|z\|$ and
$\|\bar h\|\le\ell_f\|z\|$, $m=\|z\|+\|q\|$ satisfies
$\dot m\le\bar Cm$. Since $m(0)\le\sqrt{1+(1+\alpha)^2}\,\|\zeta(t_k)\|$
and $\|\zeta\|\le(1+\alpha)m$, the claim follows.
\end{proof}

\begin{lem}[Decomposition of $\dot V_\varepsilon$]\label{lem:decomp}
For a fixed direction, along \eqref{eq:err_dyn},
\begin{equation}\label{eq:Vdot_decomp}
 \dot V_\varepsilon=\bar g(\zeta)-\beta d(1-\varepsilon\alpha)\,
 a^\top\hat La+S_u(\zeta),
\end{equation}
where
\begin{align*}
 \bar g&=-\alpha z^\top\bar h-\varepsilon w_2^\top
 \big(\alpha h_2+\beta(\hat L\otimes I_d)z_2\big)
 -\varepsilon\|w_2\|^2\\
 &\le-c_0\|\zeta\|^2,\\
 S_u&=d\,b^\top a-w_2^\top z_2-\varepsilon\beta\big(d\,b^\top\hat La
 -w_2^\top(\hat L\otimes I_d)z_2\big),
\end{align*}
and $\mathbb E[S_{u_k}(\zeta(t_k))\mid\mathcal F_k]=0$.
\end{lem}
\begin{proof}
With $\nabla_{z_1}V_\varepsilon=z_1$,
$\nabla_{z_2}V_\varepsilon=z_2+\varepsilon w_2$, and
$\nabla_{w_2}V_\varepsilon=\tfrac1{\alpha\beta}(\hat L^{-1}\otimes I_d)w_2
+\varepsilon z_2$, a direct expansion along \eqref{eq:err_dyn} using
$(\hat L\otimes P)z_2=d(\hat La)\otimes\hat u$ and
$w_2^\top(\hat L^{-1}\hat L\otimes\hat u\hat u^\top)z_2=b^\top a$ gives
\eqref{eq:Vdot_decomp}. The bound on $\bar g$ follows from
$z^\top\bar h=\tilde x^\top h\ge m_f\|z\|^2$,
$\|\alpha h_2+\beta(\hat L\otimes I_d)z_2\|\le B\|z\|$, Young's
inequality, and \eqref{eq:epsilon_choice}. The zero mean follows from
\eqref{eq:proj_moment}, since $\mathbb E[d\,b^\top a]=w_2^\top z_2$
and $\mathbb E[d\,b^\top\hat La]=w_2^\top(\hat L\otimes I_d)z_2$.
\end{proof}

The bound $\bar g\le-c_0\|\zeta\|^2$ holds along every trajectory, not
only in expectation, and the second term of \eqref{eq:Vdot_decomp} is
dissipative. Only $S_u$ has an indefinite sign. Because it is quadratic,
differentiating along \eqref{eq:err_dyn}--\eqref{eq:proj_dyn} and
dropping nonpositive terms~gives
\begin{align}\label{eq:Sdot}
 \dot S_u\le{}&\alpha\beta d^2a^\top\hat La
 +\beta\kappa\sqrt{\lambda_N}\,d^2\|b\|\sqrt{a^\top\hat La}\nonumber\\
 &+d\big(\alpha\kappa\|b\|\|h^\parallel\|
 +\beta\kappa\lambda_N\|b\|\|a\|
 +\varepsilon\beta\lambda_N\|b\|^2\nonumber\\
 &\quad+\varepsilon\alpha\beta^2\lambda_N^2\|a\|^2\big)
 +c_7\|\zeta\|^2,
\end{align}
with $c_7=1+\alpha\kappa\ell_f/2$. The two $O(d^2)$ terms are paired
with $a^\top\hat La$ and can be absorbed by the dissipative term in
\eqref{eq:Vdot_decomp}. Every other factor $d$ multiplies a product of
two projected quantities, and the next lemma shows that each has mean
square $O(1/d)$.

\begin{lem}[Projected moments]\label{lem:moments}
Let $\rho=d\tau$ and suppose
$d\tau^2\le\min\{1,1/(\alpha\beta\lambda_N(1+\ell_f))\}$. Then
\begin{equation}\label{eq:proj_sup}
 \mathbb E\Big[\sup_{t_k\le t\le t_{k+1}}\|y(t)\|^2\,\Big|\,\mathcal F_k\Big]
 \le\frac{\kappa_y}{d}\,\|\zeta(t_k)\|^2
\end{equation}
for $y\in\{a,b,h^\parallel\}$ (writing $\kappa_h$ for
$\kappa_{h^\parallel}$), and
\begin{equation}\label{eq:zeta_dev}
 \mathbb E\Big[\sup_{t_k\le t\le t_{k+1}}\|\zeta(t)-\zeta(t_k)\|^2\,\Big|\,\mathcal F_k\Big]
 \le\frac{\kappa_\delta}{d}\,\|\zeta(t_k)\|^2,
\end{equation}
where, with $C_h=\sqrt{1+\alpha^2\ell_f^2}$, $\gamma=\gamma(\tau)$ from
Lemma~\ref{lem:growth}, and
$\theta=1+\alpha^2\ell_f^2\big(1+C_h^2\gamma^2/4\big)$,
\begin{align}
 \kappa_a&=16(1+\theta\tau^2),\qquad
 \kappa_b=2\big(1+\alpha^2\beta^2\lambda_N^2\rho^2\kappa_a\big),\nonumber\\
 \kappa_z&=2\big(\rho\tau C_h^2\gamma^2+\beta^2\lambda_N^2\rho^2\kappa_a\big),\qquad
 \kappa_h=2\ell_f^2(1+\kappa_z),\nonumber\\
 \kappa_\delta&=\kappa_z+\alpha^2\beta^2\lambda_N^2\rho^2\kappa_a .
 \label{eq:kappas}
\end{align}
These constants depend on $(\alpha,\beta,\lambda_N,\ell_f,\rho,\tau)$
but not on $d$.
\end{lem}
\begin{proof}
Fix a realization of $u_k$ and write $s=t-t_k\in[0,\tau]$,
$\zeta_k=\zeta(t_k)$, and likewise $z_k,a_k,b_k,h_k^\parallel$. Let
$A=\sup_{s\in[0,\tau]}\|a(s)\|$; all suprema below are over
$s\in[0,\tau]$.

\emph{Step 1: pathwise bounds in terms of $A$.} From
\eqref{eq:err_dyn}, $\|\hat La\|\le\lambda_N\|a\|$, and
$\|(\hat La)\otimes\hat u\|=\|\hat La\|$,
\begin{align}
 \|w_2(s)-w_{2,k}\|&\le\alpha\beta\lambda_N d\,sA,\label{eq:w_dev}\\
 \|z(s)-z_k\|&\le\int_0^s\!\big(\alpha\|\bar h\|+\|w_2\|\big)
 +\beta\lambda_N d\,sA\nonumber\\
 &\le sC_h\gamma\|\zeta_k\|+\beta\lambda_N d\,sA,\label{eq:z_dev}
\end{align}
where we used $\alpha\|\bar h\|+\|w_2\|\le\alpha\ell_f\|z\|+\|w_2\|
\le C_h\|\zeta\|\le C_h\gamma\|\zeta_k\|$ by Lemma~\ref{lem:growth}.
Since $\|y^\parallel\|\le\|y\|$ and
$\|h_2(s)-h_{2,k}\|\le\|\bar h(s)-\bar h_k\|\le\ell_f\|z(s)-z_k\|$,
\begin{align}
 \|b(s)\|&\le\|b_k\|+\alpha\beta\lambda_N d\,sA,\label{eq:b_path}\\
 \|h^\parallel(s)\|&\le\|h_k^\parallel\|+\ell_f\|z(s)-z_k\|
 \nonumber\\
 &\le\|h_k^\parallel\|+\ell_f sC_h\gamma\|\zeta_k\|+\beta\ell_f\lambda_N d\,sA .
 \label{eq:h_path}
\end{align}

\emph{Step 2: closing the bound on $A$.} From \eqref{eq:proj_dyn} and
$a^\top\hat La\ge0$, $\tfrac{d}{ds}\|a\|\le\alpha\|h^\parallel\|+\|b\|$
wherever $a\ne0$. Integrating and inserting \eqref{eq:b_path} and
\eqref{eq:h_path},
\begin{align*}
 A\le{}&\|a_k\|+\tau\big(\alpha\|h_k^\parallel\|+\|b_k\|\big)
 +\tfrac{\tau^2}{2}\alpha\ell_fC_h\gamma\|\zeta_k\|\\
 &+\tfrac{\tau^2}{2}\alpha\beta\lambda_N(1+\ell_f)\,d\,A .
\end{align*}
By hypothesis $\tfrac{\tau^2}{2}\alpha\beta\lambda_N(1+\ell_f)d\le\tfrac12$,
so
\begin{equation}\label{eq:A_bound}
 A\le2\Big[\|a_k\|+\tau\alpha\|h_k^\parallel\|+\tau\|b_k\|
 +\tfrac{\tau^2}{2}\alpha\ell_fC_h\gamma\|\zeta_k\|\Big].
\end{equation}

\emph{Step 3: expectations.} Write $\mathbb E_k=\mathbb E[\,\cdot\mid\mathcal F_k]$.
The quantities $a_k,b_k,h_k^\parallel$ are projections of the
$\mathcal F_k$-measurable vectors $z_{2,k},w_{2,k},h_{2,k}$ onto the
fresh direction, so by \eqref{eq:proj_moment} and
$\|h_{2,k}\|\le\ell_f\|z_k\|$,
\begin{align*}
 \mathbb E_k\|a_k\|^2=\frac{\|z_{2,k}\|^2}{d},\quad
 \mathbb E_k\|b_k\|^2=\frac{\|w_{2,k}\|^2}{d},\\
 \mathbb E_k\|h_k^\parallel\|^2\le\frac{\ell_f^2\|z_k\|^2}{d}.
\end{align*}
Squaring \eqref{eq:A_bound} with $(\sum_{i=1}^4 y_i)^2\le 4\sum_i y_i^2$
and taking $\mathbb E_k$,
\begin{align*}
 \mathbb E_kA^2\le\frac{16}{d}\Big[&\|z_{2,k}\|^2
 +\tau^2\alpha^2\ell_f^2\|z_k\|^2+\tau^2\|w_{2,k}\|^2\\
 &+\tfrac{d\tau^4}{4}\alpha^2\ell_f^2C_h^2\gamma^2\|\zeta_k\|^2\Big].
\end{align*}
Using $d\tau^4\le\tau^2$ (from $d\tau^2\le1$) and
$\|z_{2,k}\|,\|z_k\|,\|w_{2,k}\|\le\|\zeta_k\|$,
\begin{equation*}
 \mathbb E_kA^2\le\frac{16}{d}\Big[1+\tau^2\Big(1+\alpha^2\ell_f^2
 +\tfrac{\alpha^2\ell_f^2C_h^2\gamma^2}{4}\Big)\Big]\|\zeta_k\|^2
 =\frac{\kappa_a}{d}\|\zeta_k\|^2,
\end{equation*}
which is \eqref{eq:proj_sup} for $y=a$.

For $y=b$, square \eqref{eq:b_path} with $s\le\tau$ using
$(y_1+y_2)^2\le2y_1^2+2y_2^2$:
\begin{align*}
 \mathbb E_k\sup\|b\|^2&\le\frac{2\|w_{2,k}\|^2}{d}
 +2\alpha^2\beta^2\lambda_N^2d^2\tau^2\,\mathbb E_kA^2\\
 &\le\frac{2(1+\alpha^2\beta^2\lambda_N^2\rho^2\kappa_a)}{d}\|\zeta_k\|^2
 =\frac{\kappa_b}{d}\|\zeta_k\|^2 .
\end{align*}
For the state deviation, square \eqref{eq:z_dev}:
\begin{align*}
 \mathbb E_k\sup\|z(s)-z_k\|^2
 &\le2\tau^2C_h^2\gamma^2\|\zeta_k\|^2+2\beta^2\lambda_N^2d^2\tau^2\,\mathbb E_kA^2\\
 &\le\frac{\kappa_z}{d}\|\zeta_k\|^2,
\end{align*}
where we wrote $2\tau^2C_h^2\gamma^2=2\rho\tau C_h^2\gamma^2/d$.
For $y=h^\parallel$, square the first inequality in \eqref{eq:h_path}:
\begin{align*}
 \mathbb E_k\sup\|h^\parallel\|^2&\le\frac{2\ell_f^2\|z_k\|^2}{d}
 +2\ell_f^2\,\mathbb E_k\sup\|z(s)-z_k\|^2\\
 &\le\frac{2\ell_f^2(1+\kappa_z)}{d}\|\zeta_k\|^2
 =\frac{\kappa_h}{d}\|\zeta_k\|^2 .
\end{align*}
Finally, $\|\zeta(s)-\zeta_k\|^2=\|z(s)-z_k\|^2+\|w_2(s)-w_{2,k}\|^2$
exactly, so by \eqref{eq:w_dev},
\begin{align*}
 \mathbb E_k\sup\|\zeta(s)-\zeta_k\|^2
 &\le\frac{\kappa_z}{d}\|\zeta_k\|^2
 +\alpha^2\beta^2\lambda_N^2d^2\tau^2\,\mathbb E_kA^2\\
 &\le\frac{\kappa_z+\alpha^2\beta^2\lambda_N^2\rho^2\kappa_a}{d}\|\zeta_k\|^2,
\end{align*}
which is \eqref{eq:zeta_dev} with $\kappa_\delta$ as in \eqref{eq:kappas}.
\end{proof}
\end{document}